%% file: main.tex
\RequirePackage{amsmath}
\documentclass[runningheads]{llncs}
\input{PackagesNotations}

\input{encoding} 

\usepackage{bbding}

\begin{document}

\title{
Pitts and Intuitionistic Multi-Succedent:\\
Uniform Interpolation for KM
}

\author{Hugo Férée\inst{1}\orcidID{0000-0003-3103-5612}\and  
Ian Shillito\inst{2,3}\orcidID{0009-0009-1529-2679}}

\authorrunning{H. Férée and I. Shillito}

\institute{Université Paris Cité, CNRS, IRIF, F-75013, Paris, France
\and University of Birmingham, Birmingham, United Kingdom
\and University of Melbourne, Melbourne, Australia}

\maketitle
\setcounter{footnote}{0}
\begin{abstract}
\input{abstract.tex}
\end{abstract}

\section{Introduction}

\input{introduction}

\section{Preliminaries}\label{sec:prelim}
\input{preliminaries}

\section{A $\sys{G4}$ calculus for $\KM$}\label{sec:calc}
\input{SeqCalc}

\section{Uniform interpolation for $\KM$ \`{a} la Pitts}\label{sec:UI}
\input{PropQuantif}

\section{Algebraic consequence of uniform interpolation}\label{sec:coherence}
\input{Coherence}

\section{Conclusion}
\input{conclusion}

\begin{credits}
\subsubsection{\ackname} 
The authors would like to thank Rosalie Iemhoff for her encouragements in pursuing this project,
Mojtaba Mojtahedi for his counterexample of strong completeness for $\KM$, 
Rodrigo Nicolau Almeida for bringing this logic (back) to our attention,
Sam van Gool for his clarifications on the proof of uniform interpolation for $\LC$,
and Rajeev Gor\'{e} for helpful comments on a late version of our draft.
The second author has been supported by
(1) a UKRI Future Leaders Fellowship, ‘Structure vs Invariants in Proofs’, project reference MR/S035540/1,
and
(2) the Renaissance Philanthropy grant DEEPER.
\end{credits}

\bibliographystyle{splncs04}
\bibliography{bibliography}

\end{document}

%% file: PackagesNotations.tex
\usepackage{graphicx}
\usepackage{tikz}
\usetikzlibrary{positioning,arrows,calc,shapes.geometric}
\tikzset{modal/.style={>=stealth,shorten >=1pt,shorten <=1pt,auto,node distance=0.7cm, semithick}, world/.style={circle,draw,minimum size=0.2cm}, point/.style={circle,draw,inner sep=0.5mm,fill=black},point1/.style={circle,draw,inner sep=0.5mm}, reflexive above/.style={->,loop,looseness=7,in=120,out=60}, reflexive below/.style={->,loop,looseness=7,in=240,out=300}, reflexive left/.style={->,loop,looseness=7,in=150,out=210}, reflexive right/.style={->,loop,looseness=7,in=30,out=330}, itria/.style={draw,isosceles triangle,shape border rotate=270,yshift=-1.45cm, minimum height = 15mm}, empty/.style={inner sep=0.0mm } }
\usepackage{amssymb}
\usepackage{ebproof}
\usepackage{proof}
\usepackage{newunicodechar}
\usepackage{multirow}
\usepackage[hidelinks]{hyperref}
\usepackage{enumerate}
\usepackage{array}
\usepackage{todonotes}

\usepackage{etoolbox}
\newcommand{\qedsym}{\hfill$\blacksquare$}  
\AtEndEnvironment{proof}{\phantom{}\qedsym} 

\usepackage{enumitem}
\setitemize{noitemsep,topsep=2pt,parsep=2pt,partopsep=2pt}

\def\BaseUrl{https://ianshil.github.io/KM} 
\newcommand{\rocqdoc}[1]{\href{\BaseUrl/#1}{\raisebox{-0.6mm}{\includegraphics[height=0.8em]{orocql.png}}}}

\newtheorem{obs}{Observation}

\newcommand{\Bilkova}{B{\'{i}}lkov{\'{a}}}

\newcommand{\goodbox}{\hspace{-.6ex}\text{
    \tikz[baseline=-.6ex, rounded corners=.01ex, line width=.1ex]
    {\draw (-.6ex,-.6ex) rectangle (.6ex,.6ex);}}\kern.2ex}
\renewcommand{\Box}{\goodbox}

\newcommand{\goodcrossedbox}{\hspace{-.6ex}\text{
    \tikz[baseline=-.6ex, rounded corners=.01ex, line width=.1ex]
    {\draw (-.6ex,-.6ex) rectangle (.6ex,.6ex);
      \draw (-.6ex,.6ex) -- (.6ex,-.6ex);
      \draw (-.6ex,-.6ex) -- (.6ex,.6ex);}}\kern.2ex}

\newcommand{\gooddiamond}{\hspace{-.6ex}\text{
    \tikz[baseline=-.6ex, rounded corners=.01ex, rotate=45, line width=.1ex]
    {\draw (-.5ex,-.5ex) rectangle (.5ex,.5ex);}}\kern.2ex}
\renewcommand{\Diamond}{\gooddiamond}

\renewcommand{\phi}{\varphi}

\newcommand{\sys}[1]{\mathsf{#1}}
\newcommand{\IL}{\sys{IPC}} 
\newcommand{\IPC}{\sys{IPC}} 
\newcommand{\LC}{\sys{LC}}
\newcommand{\K}{\sys{K}}
\newcommand{\KM}{\sys{KM}}
\newcommand{\KMH}{\sys{KMH}}
\newcommand{\iK}{\sys{iK}}

\newcommand{\GL}{\sys{GL}}
\newcommand{\iGL}{\sys{iGL}}

\newcommand{\iSL}{\sys{iSL}}
\newcommand{\GfouriSLt}{\sys{G4iSLt}}

\newcommand{\GfourKM}{\hyperlink{seq:GfourKM}{\sys{G4KM}}}
\newcommand{\GfouriP}{\sys{G4iP}}

\newcommand{\Kax}{\hyperlink{ax:K}{\text{K}}}
\newcommand{\KMax}{\hyperlink{ax:KM}{\text{KM}}}
\newcommand{\SLax}{\hyperlink{ax:SL}{\text{SL}}}
\newcommand{\MPax}{\hyperlink{ax:MP}{\text{MP}}}
\newcommand{\PropVar}{{\sf{Prop}}}

\newcommand{\KMAlg}{\mathsf{KMAlg}}
\newcommand{\val}{v}
\renewcommand{\int}[1]{\hat{#1}}

\newcommand{\pvf}[1]{\text{Vars}\,(#1)}
\newcommand{\pvs}[1]{\text{Vars}\,(#1)}

\newcommand{\sysrule}[1]{{\scriptstyle {(\text{#1})}}}
\newcommand{\Ra}{\Rightarrow}

\newcommand{\IdP}{\hyperlink{rule:IdP}{\sysrule{IdP}}}

\newcommand{\BotL}{\hyperlink{rule:BotL}{\sysrule{$\bot$L}}}
\newcommand{\ImpR}{\hyperlink{rule:ImpR}{\sysrule{$\rightarrow$R}}}
\newcommand{\ImpL}{\hyperlink{rule:ImpL}{\sysrule{$\rightarrow$L}}}
\newcommand{\ImpLAnd}{\hyperlink{rule:ImpLAnd}{\sysrule{$\land\!\rightarrow$L}}}
\newcommand{\ImpLOr}{\hyperlink{rule:ImpLOr}{\sysrule{$\lor\!\rightarrow$L}}}
\newcommand{\ImpLImp}{\hyperlink{rule:ImpLImp}{\sysrule{$\rightarrow\!\rightarrow$L}}}
\newcommand{\ImpLBox}{\hyperlink{rule:ImpLBox}{\sysrule{$\Box\!\rightarrow$L}}}
\newcommand{\ImpLVar}{\hyperlink{rule:ImpLVar}{\sysrule{$p\!\rightarrow$L}}}
\newcommand{\contrR}{\hyperlink{rule:contrR}{\sysrule{contrR}}}
\newcommand{\contrL}{\hyperlink{rule:contrL}{\sysrule{contrL}}}
\newcommand{\ImpLDup}{\hyperlink{rule:ImpLDup}{\sysrule{$\rightarrow\rightarrow$L$_{dup}$}}}
\newcommand{\symcut}{\hyperlink{rule:symcut}{\sysrule{Cut}}}
\newcommand{\ImpLImprev}{\hyperlink{rule:ImpLImprev}{\sysrule{$\rightarrow\!\rightarrow$L$^{-1}$}}}
\newcommand{\ImpLImptwo}{\hyperlink{rule:ImpLImptwo}{\sysrule{$\rightarrow\!\rightarrow$L2}}}

\newcommand{\BoxR}{\hyperlink{rule:BoxR}{\sysrule{$\Box$R}}}
\newcommand{\OrR}{\hyperlink{rule:OrR}{\sysrule{$\lor$R}}}
\newcommand{\AndL}{\hyperlink{rule:AndL}{\sysrule{$\land L$}}}
\newcommand{\AndR}{\hyperlink{rule:AndR}{\sysrule{$\land R$}}}
\newcommand{\OrL}{\hyperlink{rule:OrL}{\sysrule{$\lor L$}}}

\newcommand{\openboxes}{\Box^{\scriptscriptstyle{-1}}}
\newcommand{\disjunion}{,}

\newcommand{\ApKM}[1]{\mathsf{A}_{p}\big(#1\big)}
\newcommand{\EpKM}[1]{\mathsf{E}_{p}\big(#1\big)}
\newcommand{\To}{\Rightarrow}

\newcommand{\Ap}{\mathsf{A_p}}
\newcommand{\Ep}{\mathsf{E_p}}
\newcommand{\callAp}[1]{\mathsf{\mathcal{A}_p}(#1)}
\newcommand{\callEp}[1]{\mathsf{\mathcal{E}_p}(#1)}

\usepackage{pdflscape}

\newcommand{\ordering}{\prec}

\newcommand{\pfreeness}{\hyperlink{pfreeness}{\emph{$p$-freeness}}}
\newcommand{\implication}{\hyperlink{implication}{\emph{implication}}}
\newcommand{\uniformity}{\hyperlink{uniformity}{\emph{uniformity}}}

\renewcommand{\uplus}{\GenericError}

%% file: encoding.tex
\newunicodechar{⊙}{\ensuremath{\odot}}
\newunicodechar{⊗}{\ensuremath{\otimes}}
\newunicodechar{□}{\ensuremath{\Box}}
\newunicodechar{Δ}{\ensuremath{\Delta}}
\newunicodechar{δ}{\ensuremath{\delta}}
\newunicodechar{θ}{\ensuremath{\theta}}
\newunicodechar{ϕ}{\ensuremath{\phi}}
\newunicodechar{ψ}{\ensuremath{\psi}}
\newunicodechar{Γ}{\ensuremath{\Gamma}}
\newunicodechar{≺}{\ensuremath{\prec}}
\newunicodechar{∖}{\ensuremath{\setminus}}
\newunicodechar{∧}{\ensuremath{\wedge}}
\newunicodechar{∨}{\ensuremath{\vee}}
\newunicodechar{⊼}{\ensuremath{\barwedge}}
\newunicodechar{⊻}{\ensuremath{\veebar}}
\newunicodechar{₁}{\ensuremath{_1}}
\newunicodechar{₂}{\ensuremath{_2}}
\newunicodechar{₃}{\ensuremath{_3}}
\newunicodechar{⊥}{\ensuremath{\bot}}
\newunicodechar{⊤}{\ensuremath{\top}}
\newunicodechar{∈}{\ensuremath{\in}}
\newunicodechar{⇁}{\ensuremath{\rightharpoondown}}
\newunicodechar{∀}{\ensuremath{\forall}}
\newunicodechar{∃}{\ensuremath{\exists}}
\newunicodechar{λ}{\ensuremath{\lambda}}
\newunicodechar{⊢}{\ensuremath{\vdash}}
\newunicodechar{≡}{\ensuremath{\equiv}}
\newunicodechar{⋀}{\ensuremath{\bigwedge}}
\newunicodechar{⋁}{\ensuremath{\bigvee}}
\newunicodechar{∅}{\ensuremath{\emptyset}}
\newunicodechar{⊎}{\ensuremath{\uplus}}

%% file: abstract.tex
Pitts' proof-theoretic technique for uniform interpolation,
which generates uniform interpolants from terminating sequent calculi,
has only been applied to logics on an intuitionistic basis 
through single-succedent sequent calculi.
We adapt the technique to the intuitionistic multi-succedent setting 
by focusing on the intuitionistic modal logic KM. 
To do this, we design a novel multi-succedent sequent calculus for this logic
which terminates, eliminates cut, and provides a decidability argument.
Then, we adapt Pitts' technique to our calculus
to construct uniform interpolants for KM,
while highlighting the hurdles we overcame.
Finally, by (re)proving the algebraisability of KM,
we deduce the coherence of the class of KM-algebras.
All our results are fully mechanised in the Rocq proof assistant,
ensuring correctness and enabling effective computation of interpolants.

%% file: introduction.tex

Uniform interpolation is a mighty form of interpolation,
ensuring for any formula $\phi$ and variable $p$ the existence 
of \emph{left} and \emph{right uniform interpolants},
respectively denoted $\forall p\phi$ and $\exists p\phi$.
Intuitively,
$\forall p\phi$ is the strongest formula without $p$ that implies $\phi$,
and $\exists p\phi$ the weakest $p$-free formula that is implied by $\phi$.
Technically, these interpolants satisfy the following,
where $\psi$ is a $p$-free formula:
\begin{center}
\begin{tabular}{c @{\hspace{1cm}} c @{\hspace{1cm}} c @{\hspace{1cm}} c}
   $\vdash\forall p\phi\to\phi$  &
   $\dfrac{\vdash\psi\to\phi}{\vdash\psi\to\forall p\phi}$ &
   $\vdash\phi\to\exists p\phi$  &
   $\dfrac{\vdash\phi\to\psi}{\vdash\exists p\phi\to\psi}$ \\
\end{tabular}
\end{center}
Uniform interpolants are propositional formulas,
but their notation is suggestive:
they provide an interpretation of propositional quantifiers inside the logic.


Because of its strength, uniform interpolation is a notoriously difficult property to prove.
Still, a variety of proof techniques are presented in the literature:
model-theoretic~\cite{Vis1996},
universal-algebraic~\cite{GhiZaw97,GooMetTsi2017,KowMet2019},
and proof-theoretic~\cite{Pit92}.
The latter kind of technique,
developed in 1992 for intuitionistic logic $\IL$ by Pitts,
requires a \emph{terminating} sequent calculus: 
a calculus whose naive backward proof search,
i.e.~the process of repetitively applying backward rules 
of the calculus in no specific order,
necessarily comes to a halt.
The proof search tree of a sequent in a terminating calculus is \emph{finite},
and hence becomes data from which uniform interpolants are computed via mutual recursion.


In his proof, Pitts used the \emph{single-succedent} terminating calculus $\GfouriP$ for $\IL$,
which was invented several times through the decades
by Vorob'ev \cite{Vor58}, Dyckhoff~\cite{Dyc92} and Hudelmaier~\cite{Hud93}.
By exploiting the existence of such calculi,
which extend $\GfouriP$ through Iemhoff's methodology~\cite{Iem22},
Pitts' original technique was recently applied to a variety of intuitionistic modal logics~\cite{Iem2019,Gie22},
including the intuitionistic provability logic $\iSL$~\cite{FerGieGooShi24} for which a 
single-succedent terminating calculus $\GfouriSLt$ was defined~\cite{ShiGieGorIem23}.
In parallel, \Bilkova~ported the technique to \emph{multi}-succedent calculi for
\emph{classical} modal logics $\K$, $\sys{T}$ and provability logic $\GL$~\cite{Bil07}.
The shift to a classical basis allowed for technical simplifications,
e.g.~the recursive definition stops being mutual: it focuses on the left uniform interpolant without needing the right one.
Following her footsteps, van der Giessen, Jalali and Kuznets extended
her approach to additional classical modal logics using multi-succedent
but richer, i.e.~nested or labelled, sequents~\cite{GieJalKuz21,GieJalKuz23,GieJalKuz25}.
As it stands, the applicability of Pitts' technique to \emph{multi}-succedent calculi for
logics on an \emph{intuitionistic} basis remains unclear.


By scanning the literature, one easily finds logics with an intuitionistic basis
requiring multi-succedent calculi.
Most famous is the G\"{o}del-Dummett logic $\LC$~\cite{Dum59},
an intermediate logic extending $\IL$ with the linearity axiom $(\phi\to\psi)\lor(\psi\to\phi)$.
This logic is given a crucially multi-succedent terminating calculus in Dyckhoff's work~\cite{Dyc99}.
While uniform interpolation is known for $\LC$, 
as it is locally finite and has Craig interpolation~\cite{GhiZaw2002,Mak77},
it could at least serve as a good example for methodological purposes.
Of yet better interest is the intuitionistic modal logic $\KM$~\cite{Mur14},
which extends $\iSL$~\cite{GieIem20} (and hence $\iK$~\cite{BozDos84} and $\IL$) with
the \emph{Kuznetsov-Muravitsky} axiom $\Box\phi\to(\psi\lor(\psi\to\phi))$.
This logic possesses several multi-succedent calculi~\cite{Dar84,CloGor15},
and lacks a (dis)proof for uniform interpolation~\cite{Kur26}.
From a mathematical viewpoint, 
the interest in $\KM$ lies in its deep ties to $\GL$ and $\IL$:
its lattice of extensions is isomorphic to the lattice of extensions of $\GL$,
and its extension with any non-modal axiom is conservative over the extension of 
$\IL$ with the same axiom.
Additionally, $\KM$ received some attention in computer science~\cite{CloGor15}
in light of its connection to Nakano's ``later" modality~\cite{Nak00},
capturing the notion of \emph{guarded recursion}~\cite{Coq94}.
Given its proximity to $\iSL$ and the recent application of Pitts' technique 
to the latter logic~\cite{FerGieGooShi24},
$\KM$ presents itself as a natural candidate for an investigation on
the applicability of this technique to a combination of multi-succedent sequents and intuitionism.


Unfortunately, existing calculi for $\KM$ are not adequate for this investigation:
the first sequent calculi for $\KM$ given by Darjania~\cite{Dar84} clearly do not terminate,
while Clouston and Gor\'{e}'s calculus~\cite[Section 4]{CloGor15} has complex rules
and uses an alternative syntax for $\KM$.
So, we provide in Section~\ref{sec:calc} a novel multi-succedent terminating calculus for $\KM$.
Our calculus $\GfourKM$ can be obtained from $\GfouriSLt$ in two steps: 
first, port $\GfouriSLt$ to a multi-succedent setting,
making it an extension of the terminating multi-succedent calculus
$\GfouriP'$ for $\IL$~\cite[Section 7]{DycNeg00};
second, modify the implication right rule,
following (Kripke) semantic intuitions,
to capture the characteristic axiom of $\KM$.
We show that naive backward proof search in $\GfourKM$ terminates,
which, together with cut elimination, provides a decidability procedure for $\KM$.
Our syntactic proof of cut elimination uses the termination measure as induction measure,
a now standard methodology for provability logics~\cite{GorRamShi21,GorShi22,ShiGieGorIem23},
but requires a non-trivial refactoring of Dyckhoff and Negri's argument
for the admissibility of contraction~\cite{DycNeg00}.


To obtain uniform interpolation for $\KM$ using $\GfourKM$, 
it remains to adapt Pitts' technique to multi-succedent sequents.
This step mainly consists of a careful rephrasing of the uniform interpolation
property for \emph{sequents} from single- to multi-succedent.
This boils down to restricting enough the properties pertaining to the 
left uniform interpolant $\forall p\phi$ to avoid the capture of
the \emph{constant domain} propositional quantifier, satisfying the axiom
$\forall p(\phi(p)\lor\psi)\to(\forall\phi(p)\lor\psi)$ when $\psi$ is $p$-free.
With this subtlety in mind, we use the adapted technique
to prove uniform interpolation for $\KM$ in Section~\ref{sec:UI}.
We expect this adaptation to be reusable at least 
for $\LC$ and $\KM_{lin}$~\cite{CloGor15}, the combination of $\KM$ and $\LC$.


We directly put our novel result to work in Section~\ref{sec:coherence}
and infer from it
the \emph{coherence} of a class of algebras corresponding to $\KM$,
a consequence of uniform interpolation,
the algebraisability of $\KM$~\cite{Mur14},
and a bridge theorem~\cite{KowMet2019} from abstract algebraic logic~\cite{Fon16}.


Our work fits in the fast-growing literature on the
mechanisation of 
proof theory~\cite{Chap10,DawGor10,DawCloGorTiu14,Lar20},
modal logic~\cite{Doc16,DocBar18,WuGor19,MagPer23,GorRamShi21,AbrDawGor21,Fro25},
intuitionistic modal logic \cite{HagKir22,Val26}
and its proof theory~\cite{GorShi22,ShiGieGorIem23,FerGieGooShi24,BilMagPer25,BilMagPerQua24}.
Indeed, all our results
- spanning axiomatic calculus, Kripke and algebraic semantics, algebraisability, sequent calculus, decidability, admissibility of cut, uniform interpolation -
are formalised in the interactive theorem prover Rocq~\cite{Rocq}:
\url{https://github.com/ianshil/KM}.
On top of ensuring the correctness of these results,
the mechanisation allows for the extraction of executable
programs, e.g.~for effectively computing interpolants.
Throughout this paper, definitions and results are accompanied by a clickable symbol~\rocqdoc~
leading to an online-readable version of their Rocq implementations and proofs.

%% file: preliminaries.tex
We introduce the syntax, axiomatic system, and Kripke semantics of $\KM$.

\subsection{Syntax}

Let~$\PropVar=\{p,q,r\dots\}$ be a countably infinite set of propositional variables on which equality is decidable. Modal formulas (\rocqdoc{KM.Syntax.syntax.html\#form}) are defined as below:
\[
  \varphi ::= p\in\PropVar \mid \bot \mid \varphi \land \varphi \mid \varphi
  \lor \varphi \mid \varphi \to \varphi \mid \Box \varphi
\]
We use the greek letters $\varphi,\psi,\chi,\delta,\dots$ for formulas and $\Gamma,\Delta,\Sigma,\Pi\dots$ for finite multisets of formulas.
We say that $\varphi$ is a \textit{boxed formula} if $\Box$ is its main connective.
We write $\pvf{\phi}$ to denote the set of all propositional variables occurring as subformulas in the formula $\phi$,
and define $\pvs{\Gamma}:=\{q\in\PropVar\mid q\in\pvf{\phi}\text{ for some }\phi\in\Gamma\}$.
The \emph{disjoint sum} of $\Gamma$ and $\Delta$ will be denoted by $\Gamma, \Delta$. Single formulas $\phi$ will also often be coerced implicitly to the multiset singleton $\{\phi\}$.
For a multiset $\Gamma$, we define the multiset $\Box\Gamma := \{\Box\varphi: \varphi\in\Gamma\}$.
By $\openboxes\Gamma$ we mean the multiset
$\{\gamma\in\Gamma\mid\gamma\text{ is not a boxed formula}\}\cup\{\gamma\mid\Box\gamma\in\Gamma\}$~(\rocqdoc{KM.Sequent.Environments.html\#open_box}).
We write $\bigvee\Gamma$ for the disjunction of all the elements in $\Gamma$~(\rocqdoc{KM.GHC.properties.html\#list_disj}), 
where $\Gamma$ is a finite multiset.

\subsection{Axiomatic system}

We introduce an axiomatic system manipulating \emph{consecutions}, i.e.~expressions of the form $\Gamma\vdash\varphi$ where $\Gamma$ is a \emph{set} of formulas.

The generalised Hilbert calculus $\KMH$ for~$\KM$ extends the one for intuitionistic logic $\IL$
with the modal axioms and rules displayed in Figure~\ref{fig:Hilbert}.
It notably extends the intuitionistic modal logic~$\iSL$~\cite{ShiGieGorIem23},
itself an extension of $\iGL$~\cite{GieIem21},
with the Kuznetsov-Muravitsky axiom~$\Box\varphi\to(\psi\lor(\psi\to\phi))$,
and therefore proves some notable axioms:
\begin{center}
\begin{tabular}{c@{\hspace{1cm}}c@{\hspace{1cm}}c}
    $(4)\;\;\Box\phi\to\Box\Box\phi$ & $\text{(GL)}\;\;\Box(\Box\phi\to\phi)\to\Box\phi$ & $\text{(CP)}\;\;\phi\to\Box\phi$ \\
\end{tabular}
\end{center}
The strength of the Completeness Principle (CP) allows us to replace the $\sysrule{Nec}$ rule
with its variant where the premise is $\emptyset\vdash\phi$, and obtain
an equivalent logic~(\rocqdoc{KM.GHC.same_calcs.html\#gKMH_id_KMH}).
We write $\Gamma\vdash_{\KMH}\varphi$ if $\Gamma\vdash\varphi$ is provable in $\KMH$.

\begin{figure}[t]
\begin{center}
\small
\begin{tabular}{r @{\hspace{0.3cm}} l@{\hspace{0.7cm}}l}
$\Kax$ & \hypertarget{ax:K}{$\Box(\varphi\rightarrow\psi)\rightarrow(\Box\varphi\rightarrow\Box\psi)$} & \multirow{3}{4em}{$
\inferLineSkip=3pt
\infer[\scriptstyle \sysrule{Nec}]{\Gamma\vdash\Box\phi}{\Gamma\vdash\phi}
$} \\

$\SLax$ & \hypertarget{ax:SL}{$(\Box\varphi\rightarrow\varphi)\rightarrow\varphi$} & \\
$\KMax$ & \hypertarget{ax:KM}{$\Box\varphi\to(\psi\lor(\psi\to\phi))$} \\
\end{tabular}
\end{center}

\begin{center}
\begin{tabular}{c@{\hspace{1cm}}c@{\hspace{0.8cm}}c}
$
\inferLineSkip=3pt
\infer[\scriptstyle\sysrule{Ax}]{\Gamma\vdash\phi}{\varphi\text{ is an instance of an axiom }}
$ & 
$
\inferLineSkip=3pt
\infer[\scriptstyle\sysrule{El}]{\Gamma\vdash\phi}{\phi\in\Gamma}
$ &
\hypertarget{ax:MP}{$
\inferLineSkip=3pt
\infer[\scriptstyle\sysrule{MP}]{\Gamma\vdash\psi}{
	\Gamma\vdash\varphi
	&
	\Gamma\vdash\varphi\rightarrow\psi}
$}
 \\
\end{tabular}
\end{center}

\caption{Generalised Hilbert calculus $\KMH$ for $\KM$ (\rocqdoc{KM.GHC.KMH.html}).}
\label{fig:Hilbert}
\end{figure}

\subsection{Kripke semantics}

We now present the Kripke semantics for $\KM$~\cite{Mur14} which we use to
provide intuitions for the rules of our sequent calculus.

The Kripke semantics of $\KM$ restricts the one for $\iSL$ on the class of models:
the class of models for $\KM$ is strictly contained in the one for $\iSL$.
The models for $\KM$ are defined below,
where $\mathcal{P}(W):=\{U \mid U\subseteq W\}$.

\begin{definition}[\rocqdoc{KM.Kripke.kripke_sem.html\#model}]
A model $\mathcal M$ is a tuple $(W,\leq,R,I)$ satisfying the following:
$(W,\leq)$ is a preordered set;
$R$ is equal to the irreflexive part of $\leq$, noted $<$, i.e.~$R =\,<$;
$R$ is converse well-founded;
and $I:\PropVar\rightarrow\mathcal P(W)$ is a \emph{persistent} interpretation function such that
$$\forall v,w\in W.\,\forall p\in\PropVar.\,(w\leq v\;\;\rightarrow\;\; w\in I(p)\;\;\rightarrow\;\; v\in I(p))$$
\end{definition}

Although the models of $\iSL$ require that $R\subseteq\,<$,
a stricter condition is imposed for $\KM$ by enforcing the equality between $R$ and $<$, thereby making $R$ and $\leq$ interdefinable.
In other words, whenever we travel along the modal relation we know that we effectively
perform a jump to a \emph{strict} intuitionistic successor.

\begin{definition}[\rocqdoc{KM.Kripke.kripke_sem.html\#forces}]
Given a model $\mathcal M=(W,\leq,R,I)$, we define the forcing relation $\mathcal M,w\Vdash \phi$ between a world $w\in W$ and a formula $\phi$ as follows:
\begin{center}
\begin{tabular}{l@{\hspace{0.3cm}}c@{\hspace{0.3cm}}l}
$\mathcal M,w\Vdash p$ & $:=$ & $w\in I(p)$\\
$\mathcal M,w\Vdash\bot$ & $:=$ & Never \\ 
$\mathcal M,w\Vdash\phi\land\psi$ & $:=$ & $\mathcal M,w\Vdash\phi$ and $\mathcal M,w\Vdash\psi$\\
$\mathcal M,w\Vdash\phi\lor\psi$ & $:=$ & $\mathcal M,w\Vdash\phi$ or $\mathcal M,w\Vdash\psi$\\
$\mathcal M,w\Vdash\phi\rightarrow\psi$ & $:=$ & for all $v\geq w$, $\mathcal M,v\Vdash\phi$ implies $\mathcal M,v\Vdash\psi$\\
$\mathcal M,w\Vdash\Box\phi$ & $:=$ & for all $v$, $wRv$ implies $\mathcal M,v\Vdash\phi$\\
\end{tabular}
\end{center}
We abbreviate the negation of $\mathcal M,w\Vdash \phi$ by $\mathcal M,w\not\Vdash \phi$,
and write $\mathcal M, w \Vdash \Gamma$ for $\forall \varphi \in \Gamma.\,\mathcal M, w \Vdash \varphi$.
We define the local consequence \emph{(\rocqdoc{KM.Kripke.kripke_sem.html\#loc_conseq})} as follows:
\begin{center}
\begin{tabular}{l@{\hspace{1cm}}c@{\hspace{1cm}}l}
$\Gamma\models\varphi$ & iff & $\forall\mathcal M.\forall w.\,(\mathcal M,w\Vdash\Gamma\;\;\;\text{implies}\;\;\;\mathcal M,w\Vdash\varphi)$\\
\end{tabular}
\end{center}
\end{definition}

As expected, the Kripke semantics for intuitionistic logic, i.e.~\emph{persistence},
is preserved in our semantics for $\KM$.

\begin{lemma}[Persistence \rocqdoc{KM.Kripke.kripke_sem.html\#Persistence}]\label{lem:pers}
For any model $\mathcal M=(W,\leq,R,I)$, formula $\varphi$ and points $w,v\in W$, if $w\leq v$ and $\mathcal M,w\Vdash\varphi$ then $\mathcal M,v\Vdash\varphi$.
\end{lemma}

\begin{obs}\label{obs:splitsucc}
A striking feature of this semantics is the ability it gives the language
to distinguish reflexive intuitionistic jumps and jumps to strict intuitionistic successors,
an impossible feat in $\IL$ or even in $\iSL$.
Indeed, by jumping to a strict successor $v$ of $w$ all boxed formulas
forced in it become \emph{unboxed} in $v$,
while this unboxing is not generally performed by jumping to $w$ itself.
Joined with persistence, this observation gives that 
$\mathcal M,w\Vdash\Gamma$ and $w<v$ entail $\mathcal M,v\Vdash\openboxes\Gamma$.
\end{obs}

While this semantics will be handy to explain the rules of our calculus,
it is not tightly corresponding to the logic:
$\KM$ and the local consequence coincide on theorems ($\emptyset\models\phi$ if and only if $\emptyset\vdash_{\KMH}\phi)$,
but $\KM$ is not \emph{strongly complete} w.r.t.~the semantics.
Indeed, we have $\Gamma\not\vdash_{\KMH} p_0$ while $\Gamma\models p_0$, where
$\Gamma$ is the infinite set $\{\Box p_{n+1}\to p_n\mid n\in\mathbb N\}$.%
\footnote{We thank Mojtaba Mojtahedi for designing, and sharing with us, this set.}
As for most provability logics, this failure of strong completeness boils down to
the tendency converse well-founded frames have to make the local consequence non-compact,
as is known for $\GL$ and $\iGL$~\cite[Section 3.3]{Ver24}.
But for the case of $\KM$ (and $\iSL$), the argument does not hold using the usual infinite set
$\{\Diamond p_0\}\cup\{\Box(p_n\to\Diamond p_{n+1})\mid n\in\mathbb N\}$.

%% file: SeqCalc.tex
In this section we introduce our \emph{multi-succedent} sequent calculus for $\KM$
and prove several major results for it:
termination of naive backward proof search,
cut elimination,
and equivalence with $\KMH$.

As is standard in structural proof theory, 
sequents are expressions of the shape $\Gamma\Ra\Delta$
where the \emph{antecedent} $\Gamma$ and the \emph{succedent} $\Delta$
are multisets of formulas.
We present the calculus $\GfourKM$ in Figure~\ref{fig:iseq-pc}.

\hypertarget{seq:GfourKM}{\begin{figure}[t]
\centering
{\small
\hypertarget{rule:BotL}{$\begin{prooftree}
\infer0[$\BotL$]{\Gamma,\bot\Ra\Delta}
\end{prooftree}$}
\quad
\hypertarget{rule:IdP}{$\begin{prooftree}
\infer0[$\IdP$]{\Gamma,p\Ra p,\Delta}
\end{prooftree}$}
\quad
\hypertarget{rule:AndL}{$\begin{prooftree}
\hypo{\Gamma,\varphi,\psi\Ra\Delta}
\infer1[$\AndL$]{\Gamma,\varphi\land\psi\Ra\Delta}
\end{prooftree}$}
\\[0.5cm]
\hypertarget{rule:AndR}{$\begin{prooftree}
\hypo{\Gamma\Ra\varphi,\Delta}
\hypo{\Gamma\Ra\psi,\Delta}
\infer2[$\AndR$]{\Gamma\Ra\varphi\land \psi,\Delta}
\end{prooftree}$}
\quad
\hypertarget{rule:OrL}{$\begin{prooftree}
\hypo{\Gamma,\varphi\Ra\Delta}
\hypo{\Gamma,\psi\Ra\Delta}
\infer2[$\OrL$]{\Gamma,\varphi\lor\psi\Ra\Delta}
\end{prooftree}$}
\quad
\hypertarget{rule:OrR}$\begin{prooftree}
\hypo{\Gamma\Ra\varphi,\psi,\Delta}
\infer1[$\scriptstyle{\OrR}$]{\Gamma\Ra\varphi\lor\psi,\Delta}
\end{prooftree}$}
\\[0.5cm]
$\hypertarget{rule:ImpLAnd}{\begin{prooftree}
\hypo{\Gamma,\varphi\rightarrow (\psi\rightarrow\chi)\Ra\Delta}
\infer1[$\scriptstyle\ImpLAnd$]{\Gamma,(\varphi\land\psi)\rightarrow\chi\Ra\Delta}
\end{prooftree}}$
\,
\hypertarget{rule:ImpLOr}{$\begin{prooftree}
\hypo{\Gamma,\varphi\rightarrow\chi,\psi\rightarrow\chi\Ra\Delta}
\infer1[$\ImpLOr$]{\Gamma,(\varphi\lor\psi)\rightarrow\chi\Ra\Delta}
\end{prooftree}$}
\quad
\hypertarget{rule:ImpLVar}$\begin{prooftree}
\hypo{\Gamma,p,\varphi\Ra\Delta}
\infer1[$\ImpLVar$]{\Gamma,p,p\rightarrow\varphi\Ra\Delta}
\end{prooftree}$
\\[0.5cm]
\hypertarget{rule:BoxR}{\begin{prooftree}
\hypo{\openboxes Γ \disjunion □ \phi  \Ra \phi } %
\infer1[$\scriptstyle{\BoxR}$]{Γ \Ra\Box\phi,\Delta} %
\end{prooftree}}
\quad
\hypertarget{rule:ImpLBox}{\begin{prooftree}\hypo{\openboxes Γ \disjunion  \Box\phi \disjunion  \psi \Ra \phi }
\hypo{ Γ \disjunion  \psi \Ra \Delta }
\infer2[$\ImpLBox$]{Γ \disjunion  \Box\phi \rightarrow \psi \Ra \Delta}
\end{prooftree}}
\\[0.5cm]
\hypertarget{rule:ImpR}{$\begin{prooftree}
\hypo{\Gamma,\varphi\Ra\psi,\Delta}
\hypo{\openboxes\Gamma,\varphi\Ra\psi}
\infer2[$\ImpR$]{\Gamma\Ra\varphi\rightarrow\psi,\Delta}
\end{prooftree}$}
\\[0.5cm]
\hypertarget{rule:ImpLImp}{$\begin{prooftree}
\hypo{\Gamma,\chi\rightarrow\psi,\varphi\Ra\Delta, \chi}
\hypo{\openboxes\Gamma,\chi\rightarrow\psi,\varphi\Ra\chi}
\hypo{\Gamma,\psi\Ra\Delta}
\infer3[$\ImpLImp$]{\Gamma,(\varphi\rightarrow\chi)\rightarrow\psi\Ra\Delta}
\end{prooftree}$}
\caption[]{The sequent calculus~$\GfourKM$ (\rocqdoc{KM.Sequent.Sequents.html\#Provable}).}
  \label{fig:iseq-pc}
\end{figure}}
The modal rules of $\GfourKM$ are straightforwardly adapted from $\GfouriSLt$~\cite{ShiGieGorIem23}
to the multi-succedent context, 
and most non-modal rules are directly taken from $\GfouriP'$.
See in particular~\cite[Section 2.4]{ShiGieGorIem23} for an explanation of the rule~$\ImpLBox$.
Still, two essential - and interdependent - features of our calculus need commenting:
the presence of multisets on the right-hand side of sequents,
and the surprising shape of the implication right rule.
Both are targeted at proving the axiom $\KMax$
and hence the sequent $\Box p\Ra q\lor(q\to p)$.

Without a rule for disjunction on the right
preserving both disjuncts we cannot prove this sequent:
this calls for multi-succedent sequents,
giving the rule $\OrR$ of our calculus.
Such sequents can be used for $\IL$, as first shown by Maehara~\cite{Mae54},
and popularised by Dragalin~\cite{Drag88}.
Furthermore, this approach can be ported to the terminating calculus $\GfouriP$ to
obtain a multi-succedent terminating calculus $\GfouriP'$ for $\IL$~\cite[Section 7]{DycNeg00}.
This calculus has been used as basis for intuitionistic modal logics (with diamond)~\cite{DalGreOli21}, 
an example we follow for our calculus by straightforwardly adapting $\GfouriSLt$'s modal rules~\cite{ShiGieGorIem23} to multi-succedent sequents,
and providing modifications to the rules $\ImpR$ and $\ImpLImp$.

The multi-succedent $\OrR$ reduces $\Box p\Ra q\lor(q\to p)$ to $\Box p\Ra q,q\to p$, but
the implication right rule of $\GfouriP'$ presented below
does not help prove the latter sequent,
as it forces us to delete $q$ on the right.
$$
\infer{\Gamma\Ra\Delta,\phi\to\psi}{\Gamma,\phi\Ra\psi}
$$
The upward deletion of $\Delta$ in this rule,
notably preventing the provability of the excluded middle $\phi\lor\neg\phi$,
is semantically justified as the rule, read upwards, 
corresponds to a jump to an \emph{arbitrary} intuitionistic successor.
We could hope for a more detailed semantic analysis in the rule,
by distinguishing \emph{reflexive} and \emph{strict-successor} jumps.
The rule below attempts at doing just this.
$$
\infer{\Gamma\Ra\Delta, \phi\to\psi}{
    \Gamma,\phi\Ra\Delta, \psi
    &
    \Gamma,\phi\Ra\psi}
$$
Unfortunately, this is pointless: 
this last rule is equivalent to the previous one, as the left premise
is provable from the right one via weakening.
This issue is nothing but an expression of the well-known
inability of intuitionistic logic to syntactically distinguish 
reflexive from strict-successor jumps.

In $\KM$, the story is different:
by the semantic Observation~\ref{obs:splitsucc},
we can syntactically distinguish the two.
Our rule~$\ImpR$ leverages this insight.
$$
\infer[\ImpR]{\Gamma\Ra\Delta, \phi\to\psi}{
    \Gamma,\phi\Ra\Delta, \psi
    &
    \openboxes\Gamma,\phi\Ra\psi}
$$
The reflexive jump corresponds to the left premise,
as we preserved $\Delta$,
and the strict-successor jump corresponds to the right premise,
as we obtained $\openboxes\Gamma$ and deleted $\Delta$.
With the rule $\ImpR$, we can finally prove our sequent.
$$
\infer[\ImpR]{\Box p\Ra q,q\to p}{
    \infer[\IdP]{\Box p,q\Ra q,p}{}
    &
    \infer[\IdP]{p,q\Ra p}{}
}
$$
The modifications made to $\ImpR$ forced us to change the rule $\ImpLImp$ accordingly:
the two leftmost premises of $\ImpLImp$ are inspired from the premises of $\ImpR$.

We contrast our new calculus with the existing ones for $\KM$.
Darjania's multi-succedent calculi~\cite{Dar84} follow a similar intuition to ours (see also~\cite{SavSha21}),
splitting the implication right into two rules.
Contrarily to ours, they are not targeted at termination 
as evidenced by the explicit structural rules it contains for the one,
and the repetition of principal formulas in premises for the other.
Therefore, these calculi do not fit our bill: proving uniform interpolation \`{a} la Pitts.
Clouston and Gor\'{e}'s calculus~\cite[Section 4]{CloGor15} for $\KM$ terminates,
but their approach is different:
instead of using the expressivity of the modality to create the case distinction
in our rule $\ImpR$, they split the implication \emph{connective} into
a reflexive $\to$ and a strict successor $\twoheadrightarrow$ version.
The complexity of their rules 
and the semantic nature of their cut-free completeness proof, via countermodel construction,
pushed us to favour our calculus,
for which we could provide a syntactic proof of \emph{cut elimination}.


We can now show that naive backward proof search terminates in $\GfourKM$.
This is easily obtained by finding a well-founded ordering on sequents decreasing on each instance
of a rule from conclusion to premise.
For this, we reuse the adaptation to $\iSL$ of Dyckhoff's notion of weight of a formula~\cite{Dyc92},
and simply adapt the ordering induced by it from $\GfouriSLt$~\cite{ShiGieGorIem23} to multi-succedent sequents.

\begin{definition}[\rocqdoc{KM.Sequent.syntax_facts.html\#weight}]\label{def:weight}
The \textit{weight} $w(\varphi)$ of a formula is the sum of its size and of its number of conjunctions.
In other words, every symbol in a formula accounts for~$1$ in the weight, except for conjunctions which account for~$2$.
\end{definition}

\begin{definition}[Ordering on sequents \rocqdoc{KM.Sequent.Order.html\#env_pair_order}]\label{def:ordering}
    The ordering on sequents denoted by  $(\Gamma\Ra \Delta) \ordering (\Gamma' \Ra \Delta')$
    is defined as the Dershowitz–Manna multiset-ordering between $(\Gamma, \Delta, \Delta)$ and $(\Gamma', \Delta', \Delta')$ induced by the well-founded ordering on formulas induced by Definition~\ref{def:weight}.
\end{definition}

Concretely, formulas from a multiset can repeatedly be replaced by finitely-many strictly smaller formulas in order to obtain a smaller multiset. Such an ordering is well-known to be well-founded since the ordering on formulas is.
In the following, we implicitly use this ordering when performing a well-founded induction on a sequent.

It is easy to see that this ordering decreases between the conclusion and each premise of each rule from Figure~\ref{fig:iseq-pc} as a formula from the conclusion is either preserved in the premises, or replaced with several, strictly smaller formulas.
\begin{proposition}
    For any instance of a $\GfourKM$-rule with premises $s_1, \dots, s_n$ and conclusion $s$, 
    we have $s_i \ordering s$ for any $1 \le i \le n$.
\end{proposition}
\begin{proof}
  An interesting case is the~$\ImpLAnd$-rule:
  it requires that $w(\varphi\rightarrow(\psi\rightarrow\chi))<w((\varphi\land\psi)\rightarrow\chi)$,
  which holds as a conjunction adds~$2$ to the weight of a formula.
  Another one is the rule~$\BoxR$ which may not delete any formula from the sequent if $\Delta$ is empty,
  but moves $\Box \varphi$ to the left-hand-side of the sequent while keeping a copy of $\varphi$ on the right.
  This is taken into account in the ordering by counting twice multisets on the right, so that
  $(\{\Box \varphi\}\Ra \{\varphi\}) \ordering (\emptyset \Ra \{\Box \varphi \})$.
  \end{proof}

\begin{remark}
The multiset-ordering may be too strong of a tool here, as each formula is
only replaced by at most three smaller ones, not arbitrarily many.
Then if we define the weight of a sequent $\Gamma \Ra \Delta$ as $\Sigma\{3^{w(\varphi)}|\varphi  \in \Gamma, \Delta, \Delta\}$, where $w(\varphi)$ is the weight of a formula defined as above, we obtain 
a well-founded ordering that is strictly contained in $\ordering$,
still compatible with~$\GfourKM$
and that gives an explicit upper-bound for the depth of a proof tree.

However, the more general ordering $\prec$ (or at least the ordering above with a higher constant than~$3$)
still is useful when proving admissibility results for~$\GfourKM$, or when defining and proving uniform interpolation.
\end{remark}

Now, as the applicability of each rule is decidable and only finitely many rules can be applied on
a given sequent, the well-foundedness of the ordering ensures the decidability of provability in~$\GfourKM$.

\begin{proposition}[Decidability \rocqdoc{KM.Sequent.DecisionProcedure.html\#Provable_dec}]
    Provability in $\GfourKM$ is decidable.
\end{proposition}

We are now required to prove several admissibility lemmas for~$\GfourKM$.
Cut, in particular, will ensure that the calculus is complete for $\KM$. 


Before focusing on cut we have to prove the admissibility of contraction,
which has to be abandoned as a rule in $\GfourKM$ to ensure termination,
as done in $\GfouriP$ and $\GfouriSLt$.
In the following, we highlight differences between our calculus $\GfourKM$ and those inspired by Dyckhoff and Negri~\cite{DycNeg00} mentioned earlier.

First, let us observe the invertibility of the new rules.
\begin{proposition}\label{prop:inv}
The following rules are admissible:
\begin{center}
\begin{tabular}{c @{\hspace{0.8cm}} c}
\begin{prooftree}
\hypo{\Gamma \Ra \varphi \vee \psi,\Delta}
\infer1[$\scriptstyle{(\lor\text{R}^{-1})}$]{\Gamma\Ra \varphi, \psi,\Delta}
\end{prooftree}\label{eq:OrRrev}
 & 
\hfill%
\begin{prooftree}
\hypo{\Gamma \Ra \varphi \to \psi,\Delta}
\infer1[$\scriptstyle{(\rightarrow\text{R}^{-1})}$]{\Gamma, \varphi \Ra \psi,\Delta}
\end{prooftree}\label{eq:ImpRrevl} \\
\end{tabular}
\end{center}
\end{proposition}
In other words, $\OrR$ is invertible thanks to multi-succedent sequents, and
$\ImpR$ is here only partially invertible, in contrast with~$\GfouriP$ where it is fully invertible.

The proofs of admissibility of contraction on the left for $\GfouriP$ and $\GfouriSLt$ both follow the same pattern:
prove the admissibility of left implication $\ImpL$,
then the partial invertibility of $\ImpLImp$ up to contraction, i.e.~$\ImpLDup$,
before finally focusing on contraction on the left.
When turning to the the multi-succedent calculus $\GfouriP'$ for $\IPC$,
Dyckhoff and Negri lightly adapt the proof:
they first need to prove the admissibility of contraction on the right~$\contrR$
and then follow the pattern described above,
though they prove a strengthened lemma for left-implication.
Unfortunately, we cannot replicate here this succession
because of a circularity caused by the dependency of 
the admissibility of $\contrR$ on the admissibility of contraction on the left.
The root of this dependency lies in the rule $\ImpR$ for $\KM$: 
when proving the admissibility of $\contrR$ on $\Gamma\Ra\phi\to\psi,\phi\to\psi,\Delta$
with $\phi\to\psi$ principal in the last rule applied,
the duplicate is not simply weakened from premises to conclusion, as in $\GfouriP$,
but comes from the left premise $\Gamma,\phi\Ra\psi,\phi\to\psi,\Delta$.
To deal with this duplicate in the premise,
we apply Proposition~\ref{prop:inv} to get $\Gamma,\phi,\phi\Ra\psi,\psi,\Delta$,
and then need to contract on both left and right.
So, contraction on the right requires contraction on the left,
but also is at the start of a succession of lemmas leading to contraction on the left.
We solve this circular dependency by proving all these lemmas in a single mutual induction,
using the well-founded ordering on sequents from Definition~\ref{def:ordering}.

\begin{proposition}[\rocqdoc{KM.Sequent.SequentProps.html\#ImpL_dup_contr}] The following rules are admissible.
\label{prop:contractions}
\begin{center}
\begin{tabular}{c @{\hspace{0.8cm}} c}
\hypertarget{rule:contrR}{
$\begin{prooftree}
\hypo{\Gamma \Ra \psi, \psi,\Delta}
\infer1[$\contrR$]{\Gamma\Ra \psi,\Delta}
\end{prooftree}$}
& 
\hypertarget{rule:contrL}{
$\begin{prooftree}
\hypo{\Gamma, \psi, \psi \Ra\Delta}
\infer1[$\contrL$]{\Gamma, \psi\Ra \Delta}
\end{prooftree}$}
\\
& \\
\hypertarget{rule:ImpL}{
$\begin{prooftree}
\hypo{\Gamma\Ra\varphi,\Delta}
\hypo{\Gamma,\psi\Ra \Delta}
\infer2[$\ImpL$]{\Gamma, \varphi \rightarrow \psi\Ra\Delta}
\end{prooftree}$}     
& 
\hypertarget{rule:ImpLDup}{
$\begin{prooftree}
\hypo{\Gamma, (\varphi \rightarrow \psi)\rightarrow \chi \Ra \Delta}
\infer1[$\ImpLDup$]
    {\Gamma, \varphi, (\psi \rightarrow\chi), (\psi \rightarrow\chi) \Ra \Delta}
\end{prooftree}$}
\end{tabular}
\end{center}

\end{proposition}
\begin{proof}
  The four statements are proved simultaneously by well-founded induction on the conclusion of each rule, followed by a case analysis on the proof of the rightmost assumption. Here, we only sketch the cases that significantly differ from
  Dyckhoff and Negri's proofs and highlight the circular dependency between these properties.
  Crucially, we start by by proving~$\ImpL$, as it is used at the same level in the case of$~\ImpLDup$.
  \begin{enumerate}
    \item[$\ImpL$] When the last rule is $\ImpLImp$, we use the induction hypothesis for $\ImpL$, $\contrL$ and $\ImpLDup$.
    \item[$\contrR$] The main case is for the $\ImpR$-rule with $\psi = \psi_1 \rightarrow \psi_2$ as principal formula.
    The sequent $\Gamma \Ra \Delta, \psi_1 \rightarrow\psi_2$ is then proved by applying successively $\ImpR$, left contraction on~$\psi_1$, right contraction on~$\psi_2$ and partial invertibility of $\ImpR$ (Proposition~\ref{prop:inv}).
    \item[$\ImpLDup$] The main case is when the premise is proved by using the $\ImpLImp$ rule on the principal formula.
    Immediately applying the rule $\ImpL$ we just obtained (not the induction hypothesis!) leaves us with deriving the sequents
    $\Gamma, \varphi, (\psi \rightarrow\chi) \Ra \Delta$ which we have by assumption ; and
    $\Gamma, \varphi, (\psi \rightarrow\chi), \chi \Ra \Delta$, which we have by weakening on the left.
    \item[$\contrL$] The critical case is $\ImpLImp$, which uses the induction hypothesis for $\ImpLDup$.
  \end{enumerate}
\end{proof}

With contraction under our belt, we proceed to prove the admissibility of cut by local proof transformations, giving us cut elimination.
We follow the argument developed for terminating calculi for provability logics~\cite{GorRamShi21,GorShi22,ShiGieGorIem23},
which uses the inductive principle extracted from the well-founded ordering on sequents.

We prove the admissibility of the following additive variant of the cut rule. The multiplicative variant easily follows via weakening.
\begin{proposition}[\rocqdoc{KM.Sequent.Cut.html\#symmetric_cut}]\label{prop:cut}
    The additive cut rule is admissible:%
    
\hypertarget{rule:symcut}{$$\begin{prooftree}
  \hypo{\Gamma \Ra \Delta, \varphi}
  \hypo{\Gamma, \varphi \Ra \Delta}
  \infer2[$\symcut$]{\Gamma \Ra \Delta}
\end{prooftree}$$}
\end{proposition}
\begin{proof}
    By induction first on the weight of the cut formula, then by well-founded induction on the sequent $\Gamma\Ra \Delta$ as for $\GfouriSLt$~\cite{ShiGieGorIem23}.
\end{proof}


The results accumulated so far allow us to show that our calculus $\GfourKM$ captures exactly the logic $\KM$.
\begin{theorem}[\rocqdoc{KM.Sequent.Equiv_KMH.html\#G4KM_sound_KMH},\rocqdoc{KM.Sequent.Equiv_KMH.html\#G4KM_compl_KMH}]
If $\Gamma,\Delta,\{\phi\}$ is a multiset of formulas, then we have
\begin{center}
\begin{tabular}{c @{\hspace{0.3cm}}c @{\hspace{0.3cm}} c 
                @{\hspace{0.8cm}} c @{\hspace{0.8cm}}
                c @{\hspace{0.3cm}}c @{\hspace{0.3cm}} c }
$\vdash\Gamma \Ra \Delta$ & $\rightsquigarrow$ & $\Gamma \vdash_{\KMH}\bigvee\Delta$
& and &
$\Gamma \vdash_{\KMH}\phi$ & $\rightsquigarrow$ & $\vdash\Gamma \Ra \phi$ \\
\end{tabular}
\end{center}
\end{theorem}

\begin{proof}
For the implication on the left, 
we reason by induction on the proof of $\vdash\Gamma \Ra \Delta$.
Most rules of $\GfourKM$ are trivially simulated in $\KMH$,
with the exception of $\ImpR$ and $\ImpLImp$.
We focus on the former, as the latter is treated using a similar argument.
We assume that we have proofs of $\openboxes\Gamma,\phi\vdash\psi$ and
$\Gamma,\phi\vdash\psi\lor\bigvee\Delta$, and prove $\Gamma\vdash(\phi\to\psi)\vee\bigvee\Delta$.
The pivot of the proof is the obtaining of $\Gamma\vdash_{\KMH}\phi\vee(\phi\to\psi)$
from $\openboxes\Gamma,\phi\vdash_{\KMH}\psi$.
We first use the monotonicity of $\KMH$~(\rocqdoc{KM.GHC.logics.html\#KMH_monot})
and the deduction theorem~(\rocqdoc{KM.GHC.properties.html\#KMH_Deduction_Theorem}) to obtain
$\openboxes\Gamma,\Box(\phi\to\psi)\vdash\phi\to\psi$.
Then, we can use a lemma simulating the sequent rule $\BoxR$~(\rocqdoc{KM.Sequent.Equiv_KMH.html\#open_boxes_GL_rule})
giving us $\Gamma\vdash\Box(\phi\to\psi)$.
As $\Box(\phi\to\psi)\to(\phi\vee(\phi\to(\phi\to\psi))$ is an instance of the axiom $\KMax$,
we get $\Gamma\vdash\phi\vee(\phi\to(\phi\to\psi))$, hence $\Gamma\vdash\phi\vee(\phi\to\psi)$.
A reasoning by cases~(\rocqdoc{KM.GHC.properties.html\#ND_OrE})
and the use of $\Gamma,\phi\vdash_{\KMH}\psi\lor\bigvee\Delta$
gives us $\Gamma\vdash(\phi\to\psi)\vee\bigvee\Delta$.

Now, for the right direction assume $\Gamma \vdash_{\KMH}\phi$.
We reason by induction on the proof, essentially proving that the 
sequent calculus simulates the axiomatic step,
crucially using the admissibility of cut in the $\MPax$ case.
\end{proof}

\begin{remark}
Two universes co-exist in Rocq: 
the impredicative universe \texttt{Prop} and
the predicative universe \texttt{Type}.
Proving a statement in \texttt{Type} requires constructing a concrete proof term, which can later be
inspected to build other terms. However, a proof of a statement in \texttt{Prop} only guarantees that
the property holds, without granting the ability to inspect the proof. 
As we defined $\KMH$ in \texttt{Prop} and $\GfourKM$ in \texttt{Type},
the second part of our proof above is illegal in Rocq:
the sheer existence of an axiomatic proof cannot directly be used to
build a concrete sequent proof.
However, the decidability of $\GfourKM$ comes to the rescue by
allowing us to produce a proof from the knowledge of the existence of one.
\end{remark}

%% file: PropQuantif.tex
We recall the definition of the uniform interpolation property for a logic. 
\begin{definition} 
\label{def:UIP}
A logic~$\sys{L}$ has the \emph{uniform interpolation property} if, for every
$\sys{L}$-formula~$\phi$ and variable~$p$, there exist 
$\sys{L}$-formulas, denoted by $\forall p \phi$ and $\exists p \phi$,
satisfying the following three properties:
\begin{enumerate}
\item \emph{$p$-freeness:} $ \pvf{\exists p \phi} \subseteq  \pvf{\phi} \setminus \{ p \}$ and $ \pvf{\forall p \phi} \subseteq  \pvf{\phi} \setminus \{ p \}$,
\item \emph{implication:} $\vdash_{\sys{L}} \phi \to \exists p \phi \text{ and } \vdash_{\sys{L}} \forall p \phi \to  \phi,$ and
\item \emph{uniformity:} for each formula~$\psi$ with  $p \notin \pvf{\psi}$: 
\begin{align*}
    \vdash_{\sys{L}} \phi \to \psi \ &\text{ implies } \ \vdash_{\sys{L}} \exists p \phi \to \psi,\\
    \vdash_{\sys{L}} \psi \to \phi \ &\text{ implies } \ \vdash_{\sys{L}} \psi \to \forall p \phi.
\end{align*}
\end{enumerate}
\end{definition}

As suggested by the notation, this property is related to the interpretability in the propositional language of existential and universal propositional quantifiers. 

As we are going to prove this property for $\KM$ via its sequent calculus, we first need to 
generalise its definition to sequent calculi.

\begin{definition}
\label{def:SUIP}
A set of provable sequents, denoted~$\vdash$, has the \emph{uniform
interpolation property} if, for any sequent $\Gamma \To \Delta$ and variable
$p$, there exist formulas $\EpKM{\Gamma}$ and $\ApKM{\Gamma \To \Delta}$ such
that the following three properties hold: 
\begin{enumerate}
    \item \hypertarget{pfreeness}{\pfreeness}:\vspace{-2em}
    \begin{align*}  & \text{(a) } \pvf{\EpKM{\Gamma}} \subseteq \pvs{\Gamma} \setminus \{ p \}\\
    & \text{(b)} \pvf{\ApKM{\Gamma \To \Delta}} \subseteq \pvs{\Gamma, \Delta} \setminus \{ p \},
    \end{align*}\vspace{-1.5em}
    \item \hypertarget{implication}{\implication}: (a) $\vdash \Gamma \To \EpKM{\Gamma}$ and (b) $\vdash \Gamma, \ApKM{\Gamma \To \Delta} \To \Delta$, and
    \item \hypertarget{uniformity}{\uniformity}: for any multisets $\Pi$ and
        $\Sigma$ such that $p \notin \pvs{\Pi, \Sigma}$:\vspace{-.5em}
        \begin{align*}
       &\text{(a) if }p \notin \pvs{\Sigma}\text{ and }\vdash \Pi, \Gamma \To \Sigma\text{ then }\vdash \Pi, \EpKM{\Gamma} \To \Sigma\\
       &\text{\hypertarget{threeb}{(b)} if }\vdash \Pi, \Gamma \To \Delta\text{ then }\vdash \Pi, \EpKM{\Gamma} \To \ApKM{\Gamma \To \Delta}. \label{foo}
    \end{align*} 
\end{enumerate}
We say that a sequent calculus $\sf{S}$ has \emph{uniform interpolation} if $\vdash_{\sf{S}}$ has the uniform interpolation property.
\end{definition}

\begin{proposition}
    If $\GfourKM$ has the uniform interpolation property (Definition~\ref{def:SUIP}), then the logic $\KM$ has the uniform interpolation property (Definition~\ref{def:UIP}).
\end{proposition}
\begin{proof}
    Choose $\exists p\varphi := \EpKM{\{\varphi\}}$ and $\forall p\varphi := \ApKM{\emptyset \Ra \{\varphi\}}$
    and notice that $\Ep\emptyset$ is a tautology.
\end{proof}

\begin{remark}\label{rmk:constantdomain}
    The generalisation of the uniform interpolation property from formulas to sequents is not 
    straightforward,
    especially for multi-succedents sequents.
    Indeed one could be tempted to formulate the sub-property~\hyperlink{threeb}{3.(b)} as:
    $$\text{if }\vdash \Pi, \Gamma \To \Delta, \Sigma\text{ then }\vdash \Pi, \EpKM{\Gamma} \To \ApKM{\Gamma \To \Delta},\Sigma.$$
    However, this is \emph{a priori} a strictly stronger statement as it would entail the interpretability
    of a universal propositional quantifier satisfying the constant domain axiom~\cite{Gab81}:
    $$\forall p.(\varphi(p)\vee\psi)\to(\forall p.\varphi(p)\lor\psi)$$
    for any $p$-free formula $\psi$.
    This axiom is provable in second-order classical logic, 
    but not in the intuitionistic counterpart.
    While the left uniform interpolants of $\IL$
    provide an interpretation of this universal quantifier%
    \footnote{In fact, in $\IL$ they provide an interpretation of the stronger universal quantifier 
    which fully distributes over disjunctions~\cite[Example 10]{Pit92}.},
    our work seems to only establish the interpretation of a weaker, 
    more standard universal quantifier.
    However, we see no theoretical contradiction (yet) 
    in the possibility of interpreting stronger versions of this quantifier in $\KM$.
\end{remark}

We follow the structure of Pitts' original construction and its adaptation~\cite{Fer23} by defining both interpolants of sequents via mutual induction on the well-founded ordering on sequents.

\begin{definition}\label{def:EAKM}
    We define $\Ep(\Gamma) := \bigwedge\callEp{\Gamma}$ and $\Ap(\Gamma\Ra\Delta) := \bigvee\callAp{\Gamma\Ra\Delta}$
    where $\callEp{\Gamma}$ and $\callAp{\Gamma\Ra\Delta}$ are sets of formulas defined 
    by the table in Figure~\ref{fig:UIP_table}.
\end{definition}%

\begin{figure}[t]
\begin{center}
\scalebox{0.8}{
\begin{tabular}{| l l @{\hspace{0.5cm}}  l |}
\hline
    & $\Gamma$ matches & $\callEp{\Gamma}$ contains\\
\hline
    $(\Ep0)^*$ & $\Gamma',\bot$ & $\bot$ \\
    $(\Ep1)^*$ & $\Gamma',q$ & $\EpKM{\Gamma'} \wedge q$ \\
    $(\Ep2)^*$ & $\Gamma',\psi_1 \wedge \psi_2$ & $\EpKM{\Gamma', \psi_1, \psi_2}$ \\
    $(\Ep3)^*$ & $\Gamma',\psi_1 \vee \psi_2$ & $\EpKM{\Gamma',\psi_1} \vee \EpKM{\Gamma',\psi_2}$ \\
    $(\Ep4)$ & $\Gamma',(q \to \psi)$ & $q \to \EpKM{\Gamma',\psi}$ \\
    $(\Ep5)^*$ & $\Gamma',r, (r \to \psi)$ & $\EpKM{\Gamma',r, \psi}$ \\
    $(\Ep6)^*$ & $\Gamma',(\delta_1 \wedge \delta_2) \to \delta_3$ & $\EpKM{\Gamma',\delta_1 \to (\delta_2 \to \delta_3)}$ \\
    $(\Ep7)^*$ & $\Gamma',(\delta_1 \vee \delta_2) \to \delta_3$ & $\EpKM{\Gamma',\delta_1 \to \delta_3,\delta_2 \to \delta_3)}$ \\
    $(\Ep8)$ & $\Gamma',(\delta_1 \to \delta_2) \to \delta_3$ & \parbox[t]{10cm}{$\Box[\EpKM{\openboxes\Gamma',(\delta_2 \to \delta_3),\delta_1} \to \ApKM{\openboxes\Gamma',(\delta_2 \to \delta_3),\delta_1 \To\delta_2}] \to (\EpKM{\Gamma',\delta_3} \vee \EpKM{\Gamma',\delta_1\to\delta_3,\delta_1})$} \\
    $(\Ep9)$ & $\Gamma',\Box\psi$ & $\Box\EpKM{\openboxes\Gamma}$ \\
    $(\Ep10)$ & $\Gamma',\Box\delta_1 \to \delta_2$ & \parbox[t]{10cm}{$\Box[\EpKM{\openboxes\Gamma',\Box\delta_1,\delta_2} \to \ApKM{\openboxes\Gamma',\Box\delta_1,\delta_2 \To\delta_1}] \to 
    \EpKM{\Gamma',\delta_2}$} \\
\hline \hline 
    & $\Gamma\Ra \Delta $ matches & $\callAp{\Gamma \To \Delta}$ contains\\
\hline
    $(\Ap1)^*$ & $\Gamma',q \To \Delta$ & $\ApKM{\Gamma' \To \Delta}$ \\
    $(\Ap2)^*$ & $\Gamma',\psi_1 \wedge \psi_2 \To \Delta$ & $\ApKM{\Gamma', \psi_1, \psi_2 \To \Delta}$ \\
    $(\Ap3)^*$ & $\Gamma',\psi_1 \vee \psi_2 \To \Delta$ & $[\EpKM{\Gamma',\psi_1} \to \ApKM{\Gamma',\psi_1 \To \Delta}] \wedge [\EpKM{\Gamma',\psi_2} \to \ApKM{\Gamma',\psi_2 \To \Delta}]$ \\
    $(\Ap4)$ & $\Gamma',(q \to \psi) \To \Delta$ & $q \wedge \ApKM{\Gamma',\psi \To \Delta}$ \\

    $(\Ap5)^*$ & $\Gamma',r, (r \to \psi) \To \Delta$ & $\ApKM{\Gamma',r, \psi \To \Delta}$ \\
    $(\Ap6)^*$ & $\Gamma',(\delta_1 \wedge \delta_2) \to \delta_3 \To \Delta$ & $\ApKM{\Gamma',\delta_1 \to (\delta_2 \to \delta_3) \To \Delta}$ \\
    $(\Ap7)^*$ & $\Gamma',(\delta_1 \vee \delta_2) \to \delta_3 \To \Delta$ & $\ApKM{\Gamma',\delta_1 \to \delta_3,\delta_2 \to \delta_3 \To \Delta}$ \\
    $(\Ap8)$ & $\Gamma',(\delta_1 \to \delta_2) \to \delta_3 \To \Delta$ & \parbox[t]{10cm}{$[\EpKM{\Gamma',(\delta_2 \to \delta_3), \delta_1} \to \ApKM{\Gamma',(\delta_2 \to \delta_3), \delta_1 \To \Delta}] \wedge
    \Box [\EpKM{\openboxes\Gamma',(\delta_2 \to \delta_3), \delta_1} \to \ApKM{\openboxes\Gamma',(\delta_2 \to \delta_3), \delta_1 \To \delta_2}]\wedge
    \EpKM{\Gamma',\delta_3} \rightarrow \ApKM{\Gamma',\delta_3 \To \Delta}$} \\
    $(\Ap9)$ & $\Gamma \To \Delta, q$ & $q$ \\
    $(\Ap10)^*$ & $\Gamma, p \To \Delta, p$ & $\top$\\
    $(\Ap11)^*$ & $\Gamma \To \Delta, \phi_1 \wedge \phi_2$ & $\ApKM{\Gamma \To \Delta,\phi_1} \wedge \ApKM{\Gamma \To \Delta, \phi_2}$\\
    $(\Ap12)^*$ & $\Gamma \To \Delta, \phi_1 \vee \phi_2$ & $\ApKM{\Gamma \To \Delta, \phi_1,\phi_2}$  \\

    $(\Ap13)$ & $\Gamma \To \Delta, \phi_1 \to \phi_2$ & \parbox[l]{10cm}{
     $\square [\EpKM{\openboxes\Gamma',\phi_1} \to \ApKM{\openboxes\Gamma', \phi_1 \To \phi_2}] \wedge$ \hfill\\
     $[\EpKM{\Gamma,\phi_1} \to \ApKM{\Gamma, \phi_1 \To \Delta, \phi_2}]$}
    \\
    $(\Ap14)$ & $\Gamma \To \Delta, \Box\phi$ &
    $\Box[\EpKM{\openboxes\Gamma', \square \phi} \rightarrow\ApKM{\openboxes\Gamma', \square \phi \To \phi}]$  \\

    $(\Ap15)$ & $\Gamma',\Box\delta_1 \to \delta_2 \To \Delta$ & 
    \parbox[t]{10cm}{$\Box[\EpKM{\openboxes\Gamma',\Box\delta_1,\delta_2} \to \ApKM{\openboxes\Gamma',\Box\delta_1,\delta_2 \To\delta_1}] \wedge
    \ApKM{\Gamma',\delta_2 \To \Delta}$} \\

\hline
\end{tabular}
}%
\caption{Clauses to define $\EpKM{\Gamma}$ and $\ApKM{\Gamma \To \Delta}$.
In all clauses, $q \neq p$.} 
\label{fig:UIP_table}
\end{center}%
\vspace{-3em}
\end{figure}%

\begin{remark}
    While we show in Theorem~\ref{thm:KMSUIP} that Definition~\ref{def:EAKM} indeed defines uniform interpolants
    for $\GfourKM$, syntactically smaller formulas may also work.
    For instance, when $\Gamma$ (resp.~$\Gamma \Ra \Delta$) matches a line of the form $(\Ep n)^*$
    (resp.~$(\Ap n)^*$) in the table, then $\Ep(\Gamma)$ (resp. $\Ap(\Gamma\Ra\Delta)$) can be \emph{defined} as the right-hand side of that line. Indeed, such lines correspond to invertible rules in $\GfourKM$.
\end{remark}

The (very!) astute reader may have noticed an oddity in $(\Ap8)$
of the definition of the interpolant:
in the first line, the recursive $\Ap$ call is on $\Gamma',(\delta_2\to\delta_3),\delta_1\Ra\Delta$,
and not on the premise $\Gamma',(\delta_2\to\delta_3),\delta_1\Ra\delta_2,\Delta$ of the rule $\ImpLImp$.
This disappearance of $\delta_2$ on the right-hand side of this sequent in $(\Ap8)$
is motivated by the case of the $\ImpLImp$-rule in the proof of uniformity:
if it appeared in the recursive call, then the proof for this case would not go through.
The source of this problem is the fact that $\ImpLImp$ is the only left rule
with a premise whose right-hand side mixes formulas from the conclusion's 
left- and right-hand sides.

The following alternative to $\ImpLImp$ does not suffer from this issue.
\hypertarget{rule:ImpLImptwo}{$$\begin{prooftree}
\hypo{\Gamma,\chi\rightarrow\psi,\varphi\Ra\Delta}
\hypo{\openboxes\Gamma,\chi\rightarrow\psi,\varphi\Ra\chi}
\hypo{\Gamma,\psi\Ra\Delta}
\infer3[$\ImpLImptwo$]{\Gamma,(\varphi\rightarrow\chi)\rightarrow\psi\Ra\Delta}
\end{prooftree}$$}
Implicitly, this is the rule we use in $(\Ap8)$.
Why is it safe to do so?
Well, though we do not need to formally prove it, 
$\ImpLImp$ and $\ImpLImptwo$ seem equivalent:
the latter is easily proved admissible from the former and weakening on the right,
while the former is admissible from the latter in presence of cut.%
\footnote{Whether cut is admissible in the system with $\ImpLImptwo$ instead of $\ImpLImp$
is unclear: a straightforward adaptation of the argument for $\GfourKM$ does not seem to work.}

A property of $\ImpLImptwo$ we crucially use in the proof of uniformity
is its invertibility with respect to its left premise.
Note that through weakening it entails the same result for $\ImpLImp$.
\begin{proposition} The following rule is admissible in $\GfourKM$.
\hypertarget{rule:ImpLImprev}{$$\begin{prooftree}
\hypo{\Gamma, (\varphi \rightarrow \psi)\rightarrow \chi \Ra \Delta}
\infer1[$\ImpLImprev$]
    {\Gamma, \varphi, (\psi \rightarrow\chi) \Ra \Delta}
\end{prooftree}$$}
\label{prop:ImpLImprev}
\end{proposition}
\vspace{-1em}
A notable consequence of this result,
reinforcing our trust in the choice for $(\Ap8)$,
is that any sequent with a proof ending with $\ImpLImp$
can also be proved by immediately applying $\ImpLImptwo$.

\begin{theorem}[\rocqdoc{KM.Sequent.PropQuantifiers.html\#KM_uniform_interpolation}]\label{thm:KMSUIP}
$\KM$ has the uniform interpolation property.
\end{theorem}

\begin{proof}
It suffices to show that Definition~\ref{def:EAKM} provides us with uniform interpolants for $\GfourKM$, i.e.~prove the three properties of Definition~\ref{def:SUIP}.
The \pfreeness~property~(\rocqdoc{KM.Sequent.PropQuantifiers.html\#PropQuantProp.EA_vars}) is easily checked and the proof of \implication~(\rocqdoc{KM.Sequent.PropQuantifiers.html\#PropQuantProp.entail_correct}) does not differ much from the literature.
The main new difficulty is the proof of \uniformity~(\rocqdoc{KM.Sequent.PropQuantifiers.html\#PropQuantProp.pq_correct}), which goes by well-founded induction, in particular
in the case (b) of $\ImpLImp$:
we are given a $p$-free $\Pi$ and a proof $\pi$ of $\Pi,\Gamma\To\Delta$ whose last rule applied is~$\ImpLImp$
with principal formula $(\varphi\to\psi)\to\chi$, which can either belong to~$\Pi$ or~$\Gamma$.%
\begin{description}
    \item[If $(\varphi\to\psi)\to\chi\in\Pi$] then $\pi$ has the following shape.\\
\begin{small}
$
\infer[\ImpLImp]{\Pi',(\varphi\to\psi)\to\chi,\Gamma\To\Delta}{
    \deduce{\Pi',\psi\to\chi,\varphi,\Gamma\To\Delta, \psi}{\pi_1}
    &
    \deduce{\openboxes\Pi',\psi\to\chi,\varphi,\openboxes\Gamma\To\psi}{\pi_2}
    &
    \deduce{\Pi',\chi,\Gamma\To\Delta}{\pi_3}
}
$
\end{small}

We now need to prove the sequent $\Pi', (\varphi\to\psi)\to\chi,\EpKM{\Gamma}\To\ApKM{\Gamma\To\Delta}$.
For brevity, denote $\EpKM{\Gamma}$ by $\mathsf{E}^0_{p}$
and $\ApKM{\Gamma\To\Delta}$ by $\mathsf{A}^0_{p}$.\\
\begin{small}
$
\infer[\ImpLImp]{\Pi',(\varphi\to\psi)\to\chi, \mathsf{E}^0_{p}\To\mathsf{A}^0_{p}}{
    \Pi', \psi\to\chi,\varphi, \mathsf{E}^0_{p}\To\mathsf{A}^0_{p},\psi
    &
    \openboxes\Pi',\psi\to\chi,\varphi, \openboxes\mathsf{E}^0_{p}\To\psi
    &
    \Pi',\chi, \mathsf{E}^0_{p}\To\mathsf{A}^0_{p}
}
$
\end{small}\\
The middle premise is obtained by noticing that $\EpKM{\Gamma}$ implies $\Box\EpKM{\openboxes\Gamma}$
 using rule $(\Ep9)$ and then applying (IH) on $\pi_2$. 
 The rightmost premise is a direct application of (IH) on~$\pi_3$.
The leftmost premise however, cannot be derived from (IH) on $\pi_1$.
Indeed, case (a) is not applicable as $\Delta$ is not necessarily $p$-free; and case (b) only proves
$\Pi', \psi\to\chi,\varphi,\EpKM{\Gamma}\To\ApKM{\Gamma\To\Delta, \psi}$.
This is where Proposition~\ref{prop:ImpLImprev} crucially plays a role: from~$\pi$, we can obtain a proof $\pi_1'$ of 
$\Pi',\psi\to\chi,\varphi,\Gamma\To\Delta$. 
This allows the first premise to be proved by weakening $\psi$ on the right and applying the induction hypothesis (IH) on $\pi_1'$, whose end sequent is smaller than that of $\pi$.

\item[If $(\varphi\to\psi)\to\chi\in\Gamma$] then $\pi$ has the following shape.\\
\begin{small}
$
\infer[\ImpLImp]{\Pi, \Gamma,(\varphi\to\psi)\to\chi\To\Delta}{
    \deduce{\Pi,\Gamma',\psi\to\chi,\varphi\To\Delta, \psi}{\pi_1}
    &
    \deduce{\openboxes\Pi,\openboxes\Gamma',\psi\to\chi, \varphi\To\psi}{\pi_2}
    &
    \deduce{\Pi,\Gamma', \chi\To\Delta}{\pi_3}
}$\end{small}\\
Again, Proposition~\ref{prop:ImpLImprev} entails that there is a proof $\pi_1'$ of 
$\Pi,\Gamma',\psi\to\chi,\varphi\To\Delta$.

We now need to prove the sequent $\Pi,\EpKM{\Gamma',(\varphi\to\psi)\to\chi\To\Delta}\To\ApKM{\Gamma', (\varphi\to\psi)\to\chi\To\Delta}$. After weakening on the left, we notice an instance of $(\Ap8)$ on the right, producing three formulas which easily follow using $\ImpR$ and $\BoxR$ and the induction hypothesis
on $\pi_1'$, $\pi_2$ and $\pi_3$:\vspace{-0.7em}
\begin{gather*}
(1) \vdash\Pi,\EpKM{\Gamma',\psi\to\chi,\varphi\To\Delta}\To\ApKM{\Gamma',\psi\to\chi,\varphi\To\Delta}
\\
(2) \vdash\openboxes\Pi,\EpKM{\openboxes\Gamma',\psi\to\chi,\varphi\To\psi}\To\ApKM{\openboxes\Gamma',\psi\to\chi,\varphi\To\psi}
\\
(3) \vdash\Pi,\EpKM{\Gamma', \chi\To\Delta}\To\ApKM{\Gamma', \chi\To\Delta}.
\end{gather*}
\vspace{-3em}
\end{description}
\end{proof}
\vspace{-2em}

%% file: Coherence.tex
Uniform interpolation for $\KM$ has a direct consequence on the class
of algebras corresponding to this logic: \emph{coherence}.

\begin{definition}[\rocqdoc{KM.Algebra.KM_Algebras.html\#KMalg}]\label{def:CK-alg}
  A KM-algebra is a Heyting algebra $(A,\bot,\land,\lor,\to)$ together with a unary operator
  $\Box : A \to A$ satisfying\\
  $
    \Box\top = \top \quad
    \Box a \land \Box b = \Box(a \land b) \quad
    a\leq \Box a \quad
    \Box a \to a = a \quad
    \Box a \leq b \lor (b \to a)
  $\\
  where $\top:=\bot\to\bot$ and $a\leq b$ is notation for $a = a \land b$.
  We denote the class of KM-algebras by $\KMAlg$.
\end{definition}
Formulas are evaluated over KM-algebras via the notion of interpretation.

\begin{definition}[\rocqdoc{KM.Algebra.algebraic_semantic.html}]
  Let $A$ be a KM-algebra.
  A \emph{valuation} on $A$ is a function $\val:\PropVar\to A$.
  Given a valuation $\val$ over $A$ and a formula $\phi$,
  we recursively define the \emph{interpretation $\int{\val}(\phi)$ of $\phi$ 
  in $A$ via $\val$} as follows, where $\star\in\{\land,\lor,\to\}$: 
  \vspace{-0.39cm}
  \begin{center}
  \begin{tabular}{l c l @{\hspace{0.8cm}} l c l @{\hspace{0.8cm}} l c l @{\hspace{0.8cm}} l c l}
  $\int{\val}(p)$ & $:=$ & $\val(p)$ & $\int{\val}(\bot)$ & $:=$ & $\bot$ & $\int{\val}(\phi_1\star\phi_2)$ & $:=$ & $\int{\val}(\psi)\star\int{\val}(\chi)$ &
    $\int{\val}(\Box\psi)$ & $:=$ & $\Box\int{\val}(\psi)$ \\
  \end{tabular}
  \end{center}
  We define the algebraic entailment as follows:
    \begin{center}
    \begin{tabular}{l@{\hspace{1cm}}c@{\hspace{1cm}}l}
    $\Gamma\models_{\mathcal A}\varphi$ & iff & $\forall A\in\KMAlg.\,\forall\val.\;\;[\forall\gamma\in\Gamma.\;\int{\val}(\gamma)=\top]\;\Rightarrow\;\int{\val}(\phi)=\top$\\
    \end{tabular}
    \end{center}
\end{definition}
In contrast to the Kripke case, we have that $\KM$ coincides with its algebraic semantics~\cite[Proposition 3.7]{Mur14}: $\Gamma\vdash_{\KMH}\phi$ if and only if $\Gamma\models_{\mathcal A}\phi$ (\rocqdoc{KM.Algebra.KMH_algebraizable.html\#KMH_Alg1}).
\vspace{-0.05cm}
In fact, this coincidence is extremely strong: $\KM$ is \emph{finitely algebraisable} over KM-algebras (\rocqdoc{KM.Algebra.KMH_algebraizable.html\#KMH_Alg1},\rocqdoc{KM.Algebra.KMH_algebraizable.html\#KMH_Alg2},\rocqdoc{KM.Algebra.KMH_algebraizable.html\#KMH_Alg3},\rocqdoc{KM.Algebra.KMH_algebraizable.html\#KMH_Alg4})
in the sense of abstract algebraic logic~\cite{Fon16},
with $\{x=\top\}$ as defining equations and $\{x\leftrightarrow y\}$ as equivalence formulas.

A so-called \emph{bridge theorem}~\cite[Section 3.6]{Fon16} leverages algebraisability to
inform us that if $\KM$ satisfies uniform interpolation, then $\KMAlg$ satisfies the property of coherence~\cite{GooMetTsi2017,KowMet2019}:
any finitely generated subalgebra of a finitely presented member of $\KMAlg$ is finitely presented.
Consequently, Theorem~\ref{thm:KMSUIP} and the algebraisability of $\KM$ give us
the coherence of $\KMAlg$, a novel result.

\begin{corollary}
The class $\KMAlg$ of KM-algebras is coherent.
\end{corollary}
\vspace{-1em}

%% file: conclusion.tex
Our novel multi-succedent calculus $\GfourKM$ helped us establish uniform interpolation for $\KM$
via a porting of Pitts' technique to the multi-succedent setting.
Some difficulties had to be overcome on the way:
the semantically inspired modifications made to the rule $\ImpR$, and
the generalisation of the definition of uniform interpolation to multi-succedent sequents in an intuitionistic setting which avoids capturing the constant-domain universal quantifier.
We also drew a direct algebraic consequence of uniform interpolation by
showing that coherence holds of the class of KM-algebras.
To ensure the highest level of trust in our results,
we mechanised them all in the proof assistant Rocq.

In future work, we intend to leverage the insights gained in this paper,
both in the design of $\GfourKM$ and
in the adaptation of Pitts' technique, 
to show uniform interpolation for relevant logics:
the G\"{o}del-Dummett logic $\LC$ and
the combination of $\KM$ and $\LC$, called $\KM_{lin}$.
While uniform interpolation is only open for the latter and
not the former~\cite{Mak77,GhiZaw97},
applying the adaptation of Pitts' technique to $\LC$
can still be of interest: 
the size or production time of interpolants~\cite{CatKuiWol25} may be smaller using
the technique employed here than 
with the usual procedure exploiting the locally finiteness
and Craig interpolation of $\LC$.